\documentclass[conference,twocolumn]{IEEEtran}
\IEEEoverridecommandlockouts

\usepackage{amsmath,comment,amssymb}
\usepackage{graphicx}
\usepackage{times}
\usepackage{url}
\usepackage{comment}
\usepackage{overpic,subfigure}

\newtheorem{theorem}{Theorem}[section]
\newtheorem{lemma}[theorem]{Lemma}

\newenvironment{proof}[1][Proof]{\begin{trivlist}
\item[\hskip \labelsep {\bfseries #1}]}{\end{trivlist}}

\urldef \yahoo \url{www.yahoo.com}

\newcommand{\arrtimes}{t_{1}^{n}}
\newcommand{\deptimes}{t_{1}^{'n}}
\newcommand{\sizes}{s_{1}^{n}}
\newcommand{\aaone}{\mathcal{A}_{\text{A}}}

\newcommand{\efcfs}{\mathcal{E}^{c,\lambda_{2}}_{\mathrm{FCFS}}}
\newcommand{\etdma}{\mathcal{E}^{c,\lambda_{2}}_{\mathrm{TDMA}}}

\newcommand{\eptdma}{\mathcal{E}_{\mathrm{p-TDMA}}^{c,\lambda_{2}}}
\newcommand{\emax}{\mathcal{E}_{\mathrm{Max}}^{c,\lambda_{2}}}

\begin{document}

\title{Mitigating Timing Side Channel in Shared Schedulers
}

\author{
\IEEEauthorblockN{Sachin Kadloor$^{\dag *}$, Negar Kiyavash$^{\ddag *}$, Parv Venkitasubramaniam$^\S$}\thanks{This work was supported in part by National Science Foundation through the grant CNS-1117701}
\IEEEauthorblockA{$^\dag$ ECE Department and Coordinated Science Lab.\\
$^\ddag$ ISE Department and Coordinated Science Lab.\\
$^{*}$University of Illinois at Urbana-Champaign\\
$^{\S}$ECE department, Lehigh University\\
\{kadloor1,kiyavash\}@illinois.edu, parv.v@lehigh.edu}
}

\maketitle
\begin{abstract}
In this work, we  study information leakage in timing side channels that arise in the context of shared event schedulers. Consider two processes, one of them an innocuous process (referred to as Alice) and the other a malicious one (referred to as Bob), using a common scheduler to process their jobs. Based on when his jobs get processed, Bob wishes to learn about the pattern (size and timing) of jobs of Alice. Depending on the context, knowledge of this pattern could have serious implications on Alice's privacy and security. For instance, shared routers can reveal traffic patterns, shared memory access can reveal cloud usage patterns, and suchlike. We present a formal framework to study the information leakage in shared resource schedulers using the pattern estimation error as a performance metric. The first-come-first-serve (FCFS) scheduling policy and time-division-multiple-access (TDMA) are identified as two extreme policies on the privacy metric, FCFS has the least, and TDMA has the highest. However, on performance based metrics, such as throughput and delay, it is well known that FCFS significantly outperforms TDMA. We then derive two parametrized policies, accumulate and serve, and proportional TDMA, which take two different approaches to offer a tunable trade-off between privacy and performance.

\end{abstract}

\section{Introduction}

It has long been known that resources shared between processes lead to covert and side channels that can leak information from one process to another. A covert communication channel is one which is not normally intended to be used for communication \cite{covertchannelsguide}, infact, its existence is usually unknown to the system designer. Covert channels are typically used by a trusted insider with access to a secret piece of information to covey it to an outsider. Examples of covert channels include, embedding information in the unused header files of network protocols \cite{tcpheadercovert}, and two processes running on a computer communicating with each other through the access patterns of the shared memory \cite{coverttimingchannelsinprocessors}. In a covert channel, one process structures its use of the shared resource in a particular pattern so as to communicate secret information to another.  Covert channels have been studied extensively in the context of multi-level secure systems, where they can be used to create forbidden information flows~\cite{Mill87,Wray91}.

In contrast to a covert channel, in a side channel, one process tries to learn something about the operation of another without the latter's cooperation. Side channels, therefore, focus on information that is leaked \emph{incidentally} by a victim process, rather than explicitly coded by a sender. Examples of such channels include: an attacker non-invasively improving his odds of guessing cryptographic keys used by a crypto-system by observing its instantaneous power usage \cite{KJJ98}, making use of the fact that the power usage for a certain set of CPU operations is higher than others; an eavesdropper trying to guess the underlying communication by observing the encrypted packets flowing across a link \cite{Hintz02}, where he makes use of the fact that encryption does not alter the volume and timing of packets flowing on the link.

In this work, we consider the timing side channel that exists inside of a shared scheduler. A timing side channel is one in which information is conveyed (leaked) through the timings of various events. Schedulers are used in multi-tasking systems where they dictate how a finite resource is to be divided among several competing processes. Examples of such systems include: hardware resources (CPU, storage, buses) inside of a computer being shared among different processes, multiple network streams flowing through a common router, a cloud based shared computing infrastructure, etc.. In such systems, the \emph{quality-of-service} experienced by one user of the system is directly influenced by the activities of the other users of the system. For example, a sudden slowdown in web access speeds of one user could indicate an increase in network usage from the other users sharing the same network infrastructure. In this manner, a shared scheduler incidentally creates a timing side channel through which a malicious user could potentially learn about the activities of the other users using the system. 

  For the remainder of this work, these systems are abstracted out as a processor being shared by multiple users, as shown in Figure \ref{fig:sharedscheduler}. Jobs arrive from multiple users to the scheduler. The processor can work on only one job at a time. The scheduler queues up the incoming jobs, and dynamically decides on how the processor time gets divided amongst these competing jobs. In this manner, the delays experienced by jobs from one user are directly influenced by the number, timing, and size of the jobs issued by the other users.  
  
  \begin{figure}[t]
\begin{center}
\includegraphics[width=0.8\columnwidth]{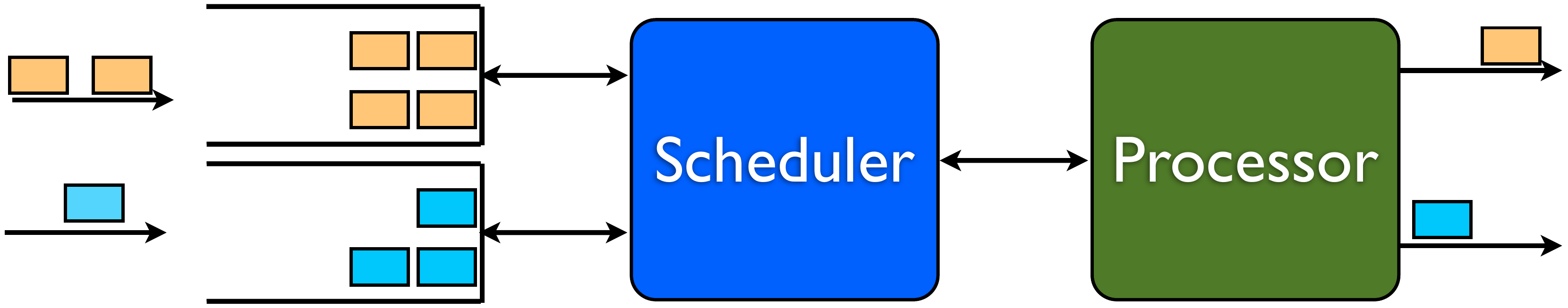}
\caption{An abstraction of a shared scheduling system}
\label{fig:sharedscheduler}
\end{center}
\end{figure}

Some of the commonly used schedulers are \emph{first-come-first-served} (FCFS): jobs are served in the order they are issued to the scheduler, \emph{time-division-multiple-access} (TDMA): each user is pre-assigned time slots during which the processor serves only jobs from that user, \emph{round-robin} (RR): one job is served from each of the queued up users in succession, \emph{priority schedulers}: one class of jobs are served ahead of another class of jobs according to some pre-defined rule, \emph{shortest-job-first} (SJF): a particular type of priority scheduler where jobs that require smaller time to be processed are served first. Each of these schedulers is optimized to perform well on a different performance metric, such as, \emph{throughput}: number of job completions per unit time, \emph{average delay}: the mean time difference between the job completion and the job arrival, \emph{fairness}: a metric to measure if the resource is being distributed equally/fairly between the processes, etc.. Like any engineered system, a system designer has to make a calculated trade-off among these conflicting metrics while picking a suitable scheduler. \emph{We argue that while choosing a scheduler that serves jobs from multiple non-trusting users, one has to consider the privacy offered by it along with the other metrics.} In this work, we explore the resulting trade-offs one has to make if privacy is taken into account.

We consider the scenario when a scheduler is serving jobs from two users, where one of them is an innocuous user and other a malicious one. The malicious user, Bob, wishes to learn the pattern of jobs sent by the innocuous user, Alice. Bob exploits the fact that when the processor is busy serving jobs from Alice, his own jobs experience a delay. As shown in Figure \ref{fig:sysmodel}, Bob computes the delays experienced by his jobs and uses these delays to infer about the times when Alice tried to access the processor, and possibly the sizes of jobs scheduled. Learning this traffic pattern from Alice can aid Bob in carrying out traffic analysis attacks.

 \begin{figure}[t] 
   \centering
   \includegraphics[width=0.6\columnwidth]{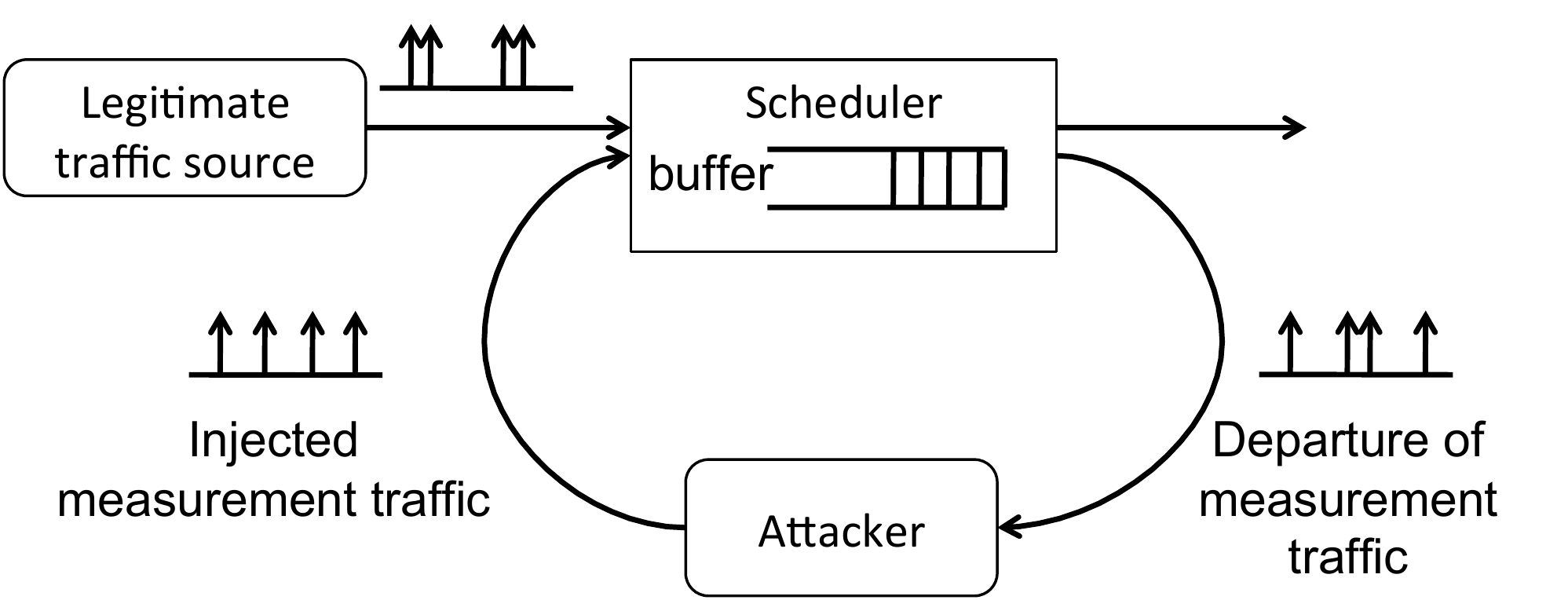}
   \caption{An event/packet scheduler being exploited by a malicious user to infer the arrival pattern from the other.}
   \label{fig:sysmodel}
\end{figure}

A summary of our main contributions follow.
\begin{enumerate}
\item \emph{Development of an analytical framework to characterize the privacy performance of a scheduling policy.} While evaluating the privacy performance of a scheduling policy, we consider the scenario described in Figure \ref{fig:sysmodel}, Alice, the innocuous user, and Bob, the malicious user, are the only two users of the system. Arrivals from Alice are modeled as a Poisson process. Bob, on the other hand, is allowed to pick the time and size of the jobs he issues to a scheduler. Furthermore, he is assumed to know the scheduling policy being used. The privacy offered by the policy is defined as the mean square error incurred by Bob when estimating Alice's arrival pattern when he picks an optimal attack strategy. Higher the error, the better the policy is at protecting the privacy of the users. 

\item \emph{Evaluation of the privacy metric of commonly deployed scheduling policies.} FCFS is one of the most commonly deployed scheduling policies owing to its simplicity. It is throughput optimal and results in minimal queuing delay. However, by the nature of the policy, there is a large correlation between the waiting times of jobs of one user and the arrival pattern of the other. Consequently, as shown in Section \ref{sec:fcfs} by the explicit construction of one attack strategy, FCFS is the weakest policy on the privacy metric. On the other hand, TDMA, wherein the delays experienced by jobs of one user are completely independent of the arrivals of the other, ranks highest on the privacy metric. However, TDMA is a highly inefficient policy in terms of throughput and delay, especially when the traffic is varying. It is especially inefficient when the number of users using the scheduler is large.
\item \emph{Design policies that offer good privacy-delay tradeoffs.} We design two parametric policies, \emph{accumulate-and-serve} and \emph{proprtional-TDMA} which can be tuned to trade-off performance for improved privacy. These policies take two different approaches to achieve privacy, one pre-distorts timing information, other pre-allocates processor times to different users. Unlike TDMA, both these policies are throughput-optimal. 
\end{enumerate}

\section{Related Works}
The current work was motivated by an earlier work of ours, \cite{icc09sachin}, wherein, we demonstrated that a side channel does exist in DSL routers, wherein a network path is shared between all the incoming packet streams. The attack is briefly described below. The attacker sends equally spaced ping packets to Alice's DSL router from his home computer. He then observes the round trip times (RTTs) of the packets. The traffic entering Alice's computer, and Bob's RTTs are shown in Figure \ref{fig:yahoo}; as can be observed, there is a clear correlation between the two. The fact that DSL routers employed FCFS scheduling aided the attack. Such an attack gives the attacker a noisy observation of the timing and the sizes of packets entering Alice's computer. Although the contents of the packets are not revealed, learning such timing information opens up the possibility for the attacker to carry out remote traffic analysis. Some of the instances of traffic analysis include recovery of information about keystrokes typed \cite{song+:sec01,sec09}, websites visited~\cite{liberatore-levine:ccs06,bissias+:pet05}, or words spoken over VoIP~\cite{spotmeifyoucanoak}. In all these works, the attacker observes Alice's traffic and uses statistical inference techniques to carry out the attack.  In \cite{clogtor}, the authors consider the scenario where a client is connected to a rogue website using a TOR network, which is designed to protect the identity of the users. The website modulates the traffic sent to the client. The website can then try to simultaneously send data through each of the TOR nodes and measure the delay incurred. By correlating this delay with the traffic it sent to the client, the website can obtain its identity, thus defeating the purpose of TOR. While that attack is no longer viable \cite{newattacktor}, the reason is that there are many more TOR nodes now than they were when \cite{clogtor} was published, and not because the timing based side channel has been eliminated. In \cite{xun_pets}, the authors exploit the side channel in a DSL router to infer the website being visited by the victim.

The increasing interest in cloud computing, wherein users issue jobs to a shared computing platform, opens up possibilities for such an attack, and is another motivation to study these timing side channels. In \cite{cloud}, the authors map the internal infrastructure of Amazon's EC2 cloud computing service, and demonstrate that it is possible for an attacker to place his virtual machine (VM) on the same physical computer as the target's VM. They shown that once placed on the same physical computer, any timing channel created by sharing of the processor can be exploited by the attacker. In cryptographic side-channels, the attacker aims at recovering cryptographic keys by utilizing the timing variations required for cryptographic operations \cite{brumley2005remote,kocher1996timing}.

While timing covert channels have been studied extensively, most notably \cite{Giles&Hajek:02IT}, there has been very little work in studying timing side channels. Most of the solutions proposed to mitigate the information leakage in side channels are system specific. Some examples are: the most common mitigation technique against such channels is cryptographic
blinding \cite{chaum1983blind,kocher1996timing}; in the context of language-based security, Agat \cite{agat2000transforming} introduces a program transformation to remove timing side channels from programs written in a sequential imperative programming language; the NRL Pump proposed for mitigating timing channels that arise in multilevel security systems (MLS) when a high confidentiality processes can communicate through Acks he sends to a low confidentiality processes \cite{Moskowitz&Miller:92IT}.

The paper is organized as follows. In Section \ref{sec:sysmodel}, we formally introduce a system model, and the metric of performance that we use to compare the privacy of different scheduling policies. In Section \ref{sec:unif_upp}, we quantify the highest degree of privacy that any scheduling policy can guarantee, and demonstrate that TDMA scheduling policy provides this. We start in Section \ref{sec:fcfs} discussing an attack strategy against FCFS policy to demonstrate that it does leak significant information between the flows. In Section \ref{sec:securepolicies}, we discuss two strategies for designing provably secure scheduling policies. Finally, in Section \ref{sec:delay}, we derive the mean delay experienced by a job when each of these policies is used, and comment on the delay-privacy tradeoff offered by the different policies.

\begin{figure}[t]
\centering
    \subfigure[DSL traffic entering Alice's router]{
    \includegraphics[width=0.39\textwidth]{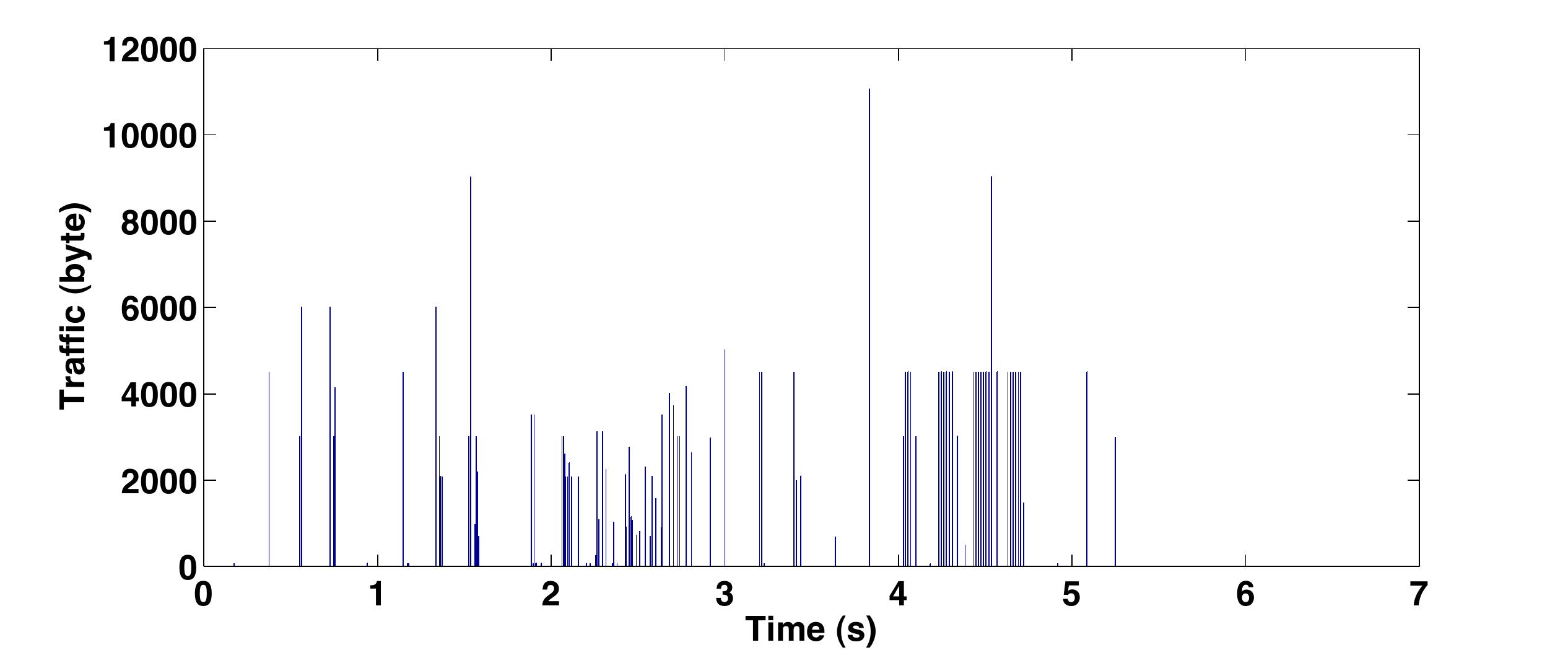}
    \label{fig:traffic}
    }
    \subfigure[Round trip times of Bob's ping packets]{
        \includegraphics[width=0.39\textwidth]{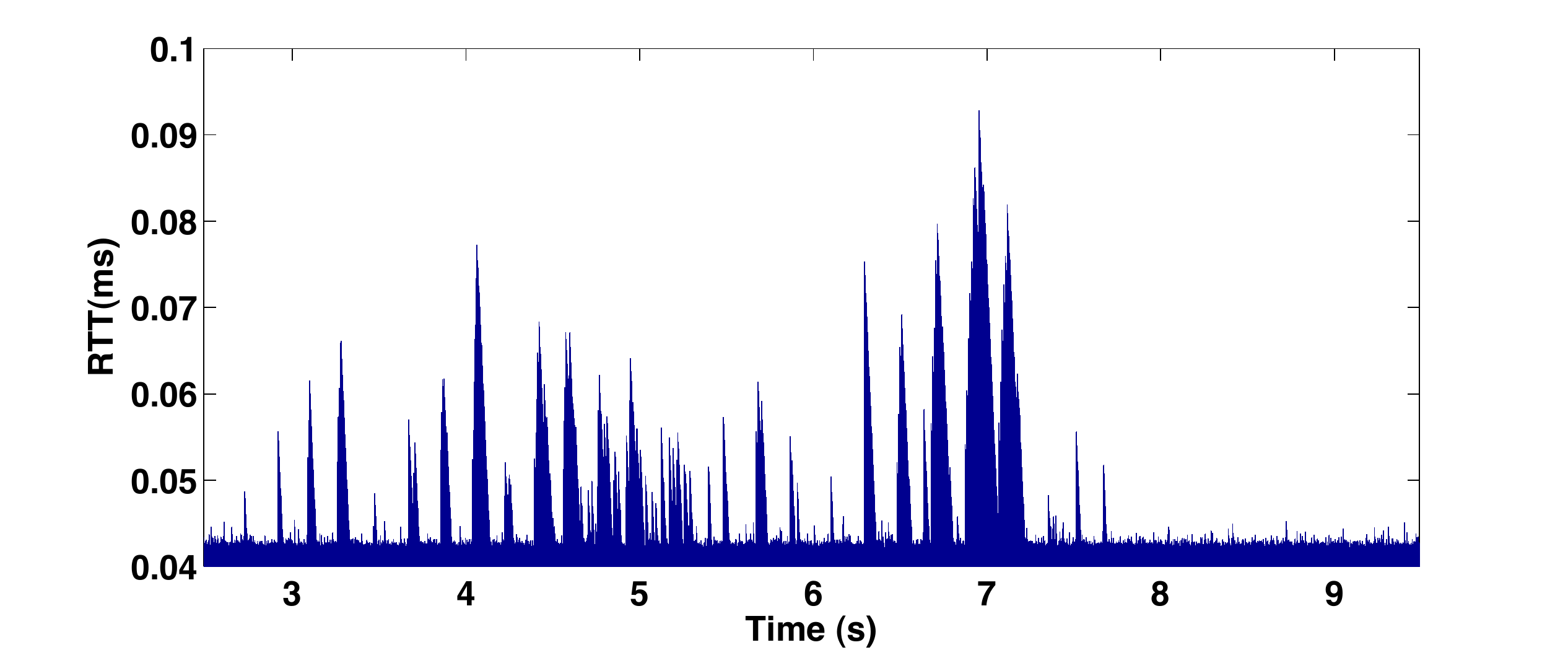}
        \label{fig:rtt}
        }
    \caption{Real traffic on a DSL line vs. observed probe RTTs when Alice is browsing the website \yahoo}\label{fig:yahoo}\vspace{-0.2in}
\end{figure}

\section{System Model and Definitions}\label{sec:sysmodel}
In this section, we introduce formally a model of the  scheduling system, and introduce a metric which measures the strength of the policy in preserving the privacy of the users. The scheduler is modeled as an infinite buffer server that is serving jobs from two users. The scheduler can serve jobs at a rate of one per unit time. We consider the scenario when one of them is an innocuous user and other a malicious one. However, the scheduler is unaware of who is malicious and who is innocuous; therefore any policy it picks should not distinguish between the users. The malicious user, Bob, wishes to exploit the queuing side channel described earlier to learn about the pattern of jobs sent by the innocuous user, Alice. Bob is assumed to know accurately the time when his jobs are issued, and the time it took for the scheduler to process it, i.e. the difference between the completion time of the job and the time when it was issued. Knowing the delays experienced by his jobs, Bob uses this information to guess the arrival pattern of jobs from Alice. 

The ability of Bob to successfully learn about Alice's arrival process depends heavily on Alice's arrival process itself. For example, on-off  patterns are easier to detect reliably compared to an arrival process that is less bursty. In order to ensure that the scheduling policies we design are robust to a variety of arrival patterns, Alice's arrival process will be modeled as a Poisson process of rate $\lambda_{2}$, with all the jobs of unit size. We do this partly because Poisson processes are known to have maximum entropy rate among processes of a given rate \cite{maxentropyrate}, and hence represent a rich class of arrival processes. Further, the analytical tractability of Poisson processes in a queuing system reveal the nature of fundamental trade-offs between privacy and delay in this system. We will comment on the case when Alice's traffic pattern follows a general arrival pattern later. Furthermore, we assume that Bob is aware of the statistical description of arrivals from Alice. A policy that guards the privacy of Alice in this scenario will also perform well when the attacker does not know a priori this rate.

\subsection{Measuring the strength of the scheduling policy: a privacy metric}
 Alice issues unit sized jobs to the scheduler according to a Poisson process of rate $\lambda_2$. The total number of jobs issued by Alice until time $u$ is given by $\aaone(u)$. The malicious user, Bob, also referred to as the attacker, issues his jobs at times $\arrtimes\doteq\{t_1,t_2,\ldots,t_n\}$, and is free to choose their sizes, $\sizes\doteq\{s_1,s_2,\ldots,s_n\}$, as well.\footnote{We have also worked on a version of this problem where the attacker is also forced to issue jobs only of unit size, refer \cite{Infocom2012_kadloor}. The present work is therefore a generalization of the former.} Let $\deptimes\doteq\{t_{1}^{'},t_{2}^{'},\ldots,t_{n}^{'}\}$ be the departure times of these jobs. Bob makes use of the observations available to him, the set $\{\arrtimes,\sizes,\deptimes\}$ and the knowledge of the scheduling policy used, in estimating Alice's arrival pattern. The arrival pattern of Alice is the sequence $\{X_k\}_{k=1,2,\ldots,N}$, where $X_k=\aaone(kc)-\aaone((k-1)c)$, is the number of jobs issued by Alice in the interval $((k-1)c,kc]$, referred to as the $k^{th}$ clock period of duration $c$. $Nc$ is the time horizon over which the attacker is interested in learning Alice's arrival pattern.

The privacy offered by a scheduling policy is measured by the long run estimation error incurred by Bob in such a scenario when he is free to choose the number of jobs he issues, times when he issues them and their sizes,  subject to a maximum rate constraint, and when he optimally estimates Alice's arrival pattern. Formally, the privacy offered by a scheduling policy is defined to be:
\begin{align}
\lefteqn{\mathcal{E}_{\mathrm{Scheduling\ policy}}^{c,\lambda_{2}} =\lim_{N\to \infty} \  \min_{n,t_{1}^{n},s_{1}^{n}:\frac{\sum\limits_{i=1}^{n}s_{i}}{Nc}<1-\lambda_{2}}}     \nonumber \\ 
& \hspace{0.75in} \sum_{k=1}^{N}  \mathbf{E}\left[\left(X_k-\mathbf{E}\left[X_k|t_{1}^{n},t^{'n}_{1},s_{1}^{n}\right]\right)^2\right], 
\label{eq:privacy}
\end{align}
where, the expectation is taken over  the joint distribution of the arrival times of Alice's jobs, the arrival times and sizes of jobs from the attacker and his departure times. This joint distribution is in turn dependent on the scheduling policy used, which is known to the attacker. Finally, the attacker is assumed to know the statistical description of Alice's arrival process, and he is allowed to pick $\sum\limits_{i=1}^{n}s_{i}/{Nc}$, the average rate at which he issues his jobs, to be any value that is less than $1-\lambda_{2}$, so as to keep the system stable. A scheduling policy is said to preserve the privacy of its users if the resulting estimation error is high.

\paragraph*{A game-theoretic viewpoint of the privacy metric}
The optimal value of the estimation error can be viewed in a game theoretic framework as a saddle point between the scheduler and Bob; Bob's task is to minimize estimation error, and the scheduler's task is to maximize it. However, we force the scheduler to \emph{play first} and announce the scheduling policy before the attacker chooses an attack strategy. This is motivated by practical considerations. In most system, the scheduling protocol is known to the users apriori. Further, in side channel attacks, it is not possible for the scheduler to identify a malicious user and tailor the policy accordingly. Consequently, this assumption empowers the attacker and the achievable privacy would be a benchmark for attackers equipped with lesser information.

Our motivation for using the minimum mean squared error as a metric of performance is as follows. The minimum mean squared error, as considered in this paper, does not conform to a specific adversarial learning technique, but serves as a universal lower bound over all adversarial strategies taking into account the complete available information for the entire duration of the system operation. A natural alternative metric would be to measure the information leakage using Shannon's equivocation $\lim \limits_{N\to \infty} \frac{1}{N} \sum_{k=1}^{N}H(X_{k}|t_{1}^{n},t^{'n}_{1},s_{1}^{n})$. While entropy serves as a measure of uncertainty, which guarantees a minimum probability of error for an adversary (Fano's inequality), MMSE bounds the actual error incurred. The purpose of quantifying privacy is to have a meaningful measure of how breachable a system is, and in that respect, both these measures provide that interpretation. Furthermore, the two metrics are related in the sense that both MMSE and entropy are functionals of the probability mass function of the same conditional random variable, $X_k|t_{1}^{n},t^{'n}_{1},s_{1}^{n}$. MMSE is the variance of the random variable, while equivocation is the entropy of the random variable. It can be shown, using proof techniques similar to those in \cite{entropy_error}, that large estimation error does correspond to large entropy, although the relationship between the two is not monotonic.  

In the subsequent section, we provide a bound on the maximum estimation error an attacker could incur, demonstrate that TDMA scheduling policy ensures that the attacker incurs this error. Subsequently, we  present a scheduling policy that achieves an estimation error close to the maximum estimation error at the cost of a limited increase in delay.
 
 \section{Performance of FCFS and TDMA on the privacy metric}\label{sec:fcfstdmaperformance}
 In this section, we prove that the FCFS and TDMA are two extreme policies on the privacy metric, FCFS offers the least privacy while TDMA offers the highest. 
 
 \subsection{An upper bound on the estimation error of the attacker}\label{sec:unif_upp}
 In order to gauge the amount of information that the attacker learns through this side channel, it is important to first study the amount of information that the attacker has even before performing any attack. In our case, the attacker is assumed to know the statistical description that governs the arrival pattern from Alice, that it is a Poisson traffic of rate $\lambda_{2}$. In this section, we compute the estimation error incurred by the attacker when he uses just this statistical description to estimate Alice's arrival pattern. This is the maximum estimation error that a rational attacker can incur, and will also serve as a benchmark to test the efficacy of a scheduling policy in preserving Alice's privacy. 
 
 \begin{theorem}
 Irrespective of the scheduling policy used, the maximum estimation error that the attacker can incur is $\lambda_{2}c$. 
 \end{theorem}
 \begin{proof}
 When Alice's traffic is a Poisson $\lambda_{2}$ traffic, the number of arrivals in between two clock ticks, $X_{k}$, is a Poisson random variable with parameter $\lambda_{2}c$. Also $X_{k}$s are independent and identically distributed (i.i.d.). Ignoring all the observations available to him, \emph{viz.}~$\{t_{1}^{n},t^{'n}_{1},s_{1}^{n}\}$, if Bob estimates $X_{k}$ with its statistical mean, $\lambda_{2}c$, for all $k$, the mean estimation error incurred by him is equal to the variance of the Poisson random variable, which is $\lambda_{2}c$. 
  \end{proof}
 Define 
 \begin{eqnarray} 
   \mathcal{E}_{\mathrm{Max}}^{c,\lambda_{2}} &\doteq& \lambda_{2}c.  \label{eq:unifub} \\
   &\geq& \mathcal{E}_{\mathrm{Scheduling\ Policy}}^{c,\lambda_{2}} \label{eq:unifbound}
 \end{eqnarray}
 Making a clever use of his observations, the attacker can potentially perform a better job at guessing Alice's arrival process, thereby incurring a lower estimation error. A good scheduling policy should have an estimation error as close $\mathcal{E}_{\mathrm{Max}}^{c,\lambda_{2}}$ as possible. 
 
 \subsection{Time Division Multiple Access (TDMA) policy achieves the maximum privacy}\label{sec:tdma}
 If the scheduler uses a Time-Division-Multiple-Access (TDMA) policy, it allocates dedicated time slots to process jobs from each user. For instance, when there are only two users, the scheduler could assign the odd time slots to serve jobs from the first user, and even time slots to process jobs from the other. If there are no unserved jobs from one of the users, the scheduler just idles in the corresponding time slots. 
 
 For such a policy, the arrival pattern of one user does not influence the completion times of jobs in the other. In this scenario, the attacker can do nothing better than use his knowledge of the Alice's arrival statistics to guess the arrival pattern. Therefore, the estimation error incurred would be $\etdma=\lambda_{2}c$. This means that from an privacy perspective, TDMA is an optimal albeit impractical scheduling policy. 
 
The high privacy offered by TDMA comes at a significant performance cost in terms of throughput and delay. First, note that the policy is not throughput optimal \cite{Rom&Sidi:Book, Gallager:book}. By throughput optimality, we mean the following. Suppose there are $m$ users using the scheduler to process their jobs, and user $i$ issues unit sized jobs at a rate $\lambda_{i}$, and the scheduler can process 1 unit of job in a unit interval of time. It can be shown that the system stays stable, i.e., the queue sizes do not blow to infinity, for any values of the input rates as long as they satisfy $\sum\limits_{i=1}^{m}\lambda_{i}<1$ if the scheduler uses the FCFS policy. However, if the scheduler used TDMA, then for the system to be stable, we would require $\lambda_{i}<1/m$, which significantly shrinks the \emph{rate region} over which the system stays stable. Second, even if the arrival rates from users are such that the system stays stable, as user's job patterns are rarely periodic, they would therefore incur severe delays. The delay issues are further discussed in Section \ref{sec:delay}.

 \subsection{FCFS offers the least privacy}\label{sec:fcfs}
 The FCFS scheduling policy serves jobs in the order in which they  arrive. Although this policy is easy to implement, and fares well in the metrics of throughput and mean delay, as we will demonstrate, it does not perform well in terms of preserving the privacy of  a user. For the system model considered here, the following theorem can be proved. It states that the estimation error incurred by the attacker when FCFS scheduling policy is used is equal to zero when the attacker issues his jobs at a high enough rate. 
 
 \begin{theorem}\label{thm:fcfszeroerr}
 FCFS policy offers no privacy to its users. Specifically, for a fixed arrival rate from Alice, $\lambda_{2}$, the estimation error incurred by the strongest attacker is equal to zero.  That is,
 \begin{eqnarray}
 \efcfs =0.
 \end{eqnarray}
 \end{theorem}
 
\begin{IEEEproof}
We prove this theorem by specifying one particular attack strategy, i.e., one set of arrival times and sizes of the jobs from the attacker which guarantees zero error. The attacker issues one job every $c/\lceil c \rceil$ time units, where $\lceil c \rceil$ is the smallest integer greater than or equal to $c$. The size of each job is $\lambda c/\lceil c \rceil $, so that the rate at which he issues jobs is equal to $\lambda$. Therefore $t_k=kc/\lceil c \rceil$ and $s_k=\lambda c/\lceil c \rceil $. Let $t_{1}^{'n}$ be the departure times of his jobs, and $\tilde{X}_k$ be the number of jobs issued by Alice in between times $(k-1) c/\lceil c \rceil $ and $k c/\lceil c \rceil$. Recall that ${X}_k$ is the number of jobs issued by Alice in between times $(k-1)c$ and $kc$. Therefore, $X_k=\sum\limits_{i=1}^{\lceil c \rceil}\tilde{X}_{(k-1)\lceil c \rceil + i }$. This means if the attacker estimates the sequence $\{\tilde{X}_k\}$ accurately, he can estimate the sequence $\{X_k\}$ accurately as well. Note that the time between successive arrivals from the attacker is $ c/\lceil c \rceil < 1 $. The size of a job from Alice is 1. Hence the attacker can accurately learn whether Alice issued a job between two of his jobs or not.  His estimation procedure is the follows. 

Consider first the scenario when $t_{k-1}^{'}<t_{k}$, that is, his $k-1^{th}$ job departs before the arrival of $k^{th}$ job. In this case, suppose $\tilde{X}_k=0$, then $t_{k}^{'}=t_{k}+s_k$, i.e., the $k^{th}$ job goes into service immediately upon its arrival. This scenario is shown in Figure \ref{fig:earlydepboth}. On the other hand, if $\tilde{X}_k>0$, shown in Figure \ref{fig:earlydepfirst} then $t_k+s_k+\tilde{X}_k-1<t_{k}^{'}<t_k+s_k+\tilde{X}_k$, which implies $\tilde{X}_k=\lceil t_{k}^{'}-(t_k+s_k) \rceil $. 

Now, suppose $t_{k-1}^{'}>t_{k}$, i.e., if the $k-1^{th}$ job from the attacker departs after the arrival of his $k^{th}$ job, shown in Figure \ref{fig:latedepboth}. In this case, $\tilde{X}_k=t_{k}^{'}-s_k-t_{k-1}^{'}$.

Clearly, $\tilde{X}_{k}$ is a deterministic function of $t_{k-1},t_{k-1}^{'},t_{k},t_{k}^{'}$ and $s_k$. Therefore, the attacker incurs zero error in estimating Alice's arrival pattern.
\end{IEEEproof}

\begin{figure}[t]
\centering
\resizebox{0.35\columnwidth}{!}{
\begin{overpic}[width=0.2\textwidth]{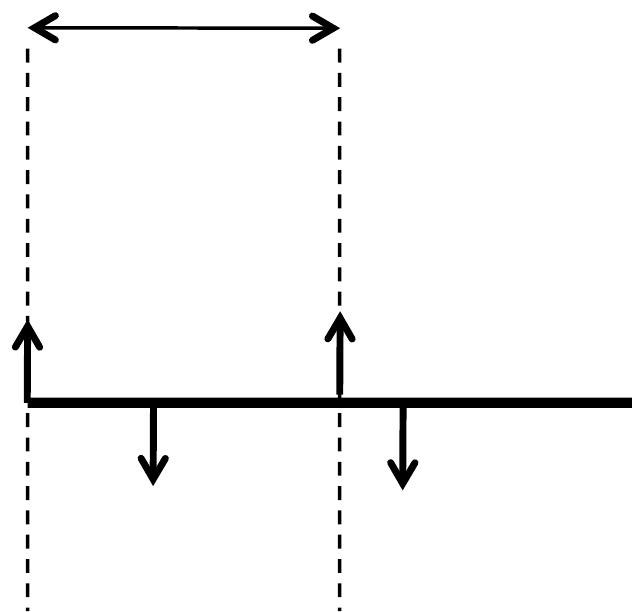}
\put(-7,-4){ $\frac{(k-1)c}{\lceil c \rceil}$}
\put(50,-4){$\frac{kc}{\lceil c \rceil}$}
\put(18,38){$t^{'}_{k-1}$}
\put(60,38){$t^{'}_{k}$}
\put(13,96){ $\tilde{X}_{k}=0$}
\end{overpic}
}
\caption{This is the scenario when $t_{k-1}^{'}<kc/\lceil c \rceil$ and $\tilde{X}_{k}=0$. In this scenario, $t_{k}^{'}=t_{k}+s_{k}$.  In the figure, the solid upward and downward pointing arrows denote the arrival and departure times, respectively, of attacker's jobs.}
\label{fig:earlydepboth}
\end{figure}

\begin{figure}[tbp]
\centering
\resizebox{0.35\columnwidth}{!}{
\begin{overpic}[width=0.23\textwidth]{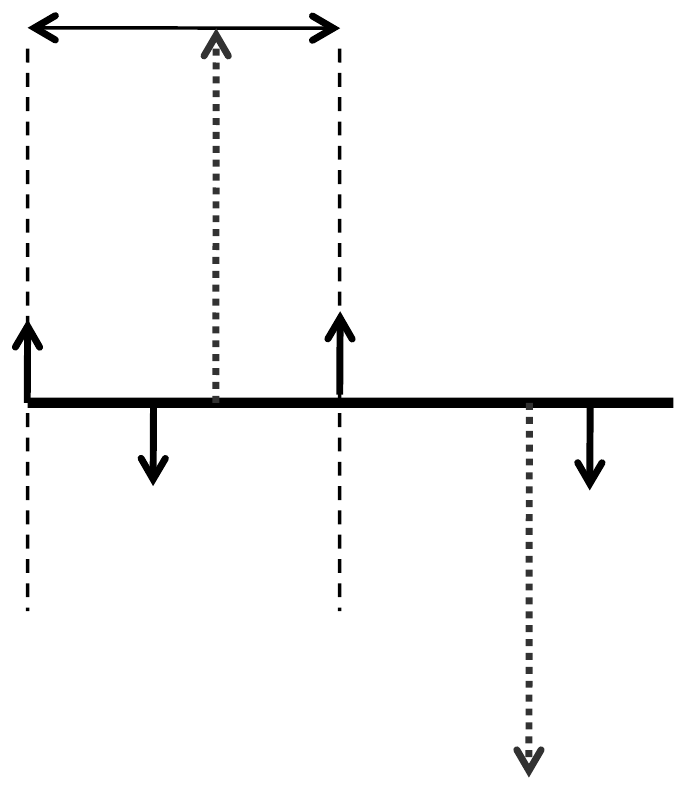}
\put(-7,14){ $\frac{(k-1)c}{\lceil c \rceil}$}
\put(38,16){$\frac{kc}{\lceil c \rceil}$}
\put(14,55){$t^{'}_{k-1}$}
\put(70,55){$t^{'}_{k}$}
\put(13,100){ $\tilde{X}_{k}=1$}
\end{overpic}
}
\caption{This is the scenario when $t_{k-1}^{'}<kc/\lceil c \rceil$ and $\tilde{X}_{k}=1>0$. In this scenario, $\tilde{X}_{k}=\lceil t_{k}^{'}-(t_{k}+s_{k}) \rceil$.  In the figure, the solid upward and downward pointing arrows denote the arrival and departure times, respectively, of attacker's jobs, and the dotted arrows represent the jobs from Alice. The size of the arrow is proportional to the size of the job. }
\label{fig:earlydepfirst}
\end{figure}

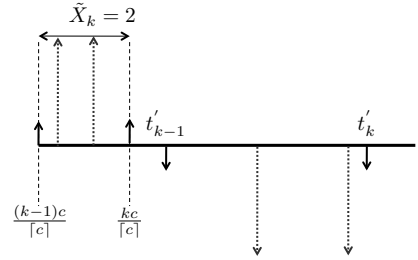
\begin{figure}[tp]
\centering
\resizebox{0.6\columnwidth}{!}{
\begin{overpic}[width=0.35\textwidth]{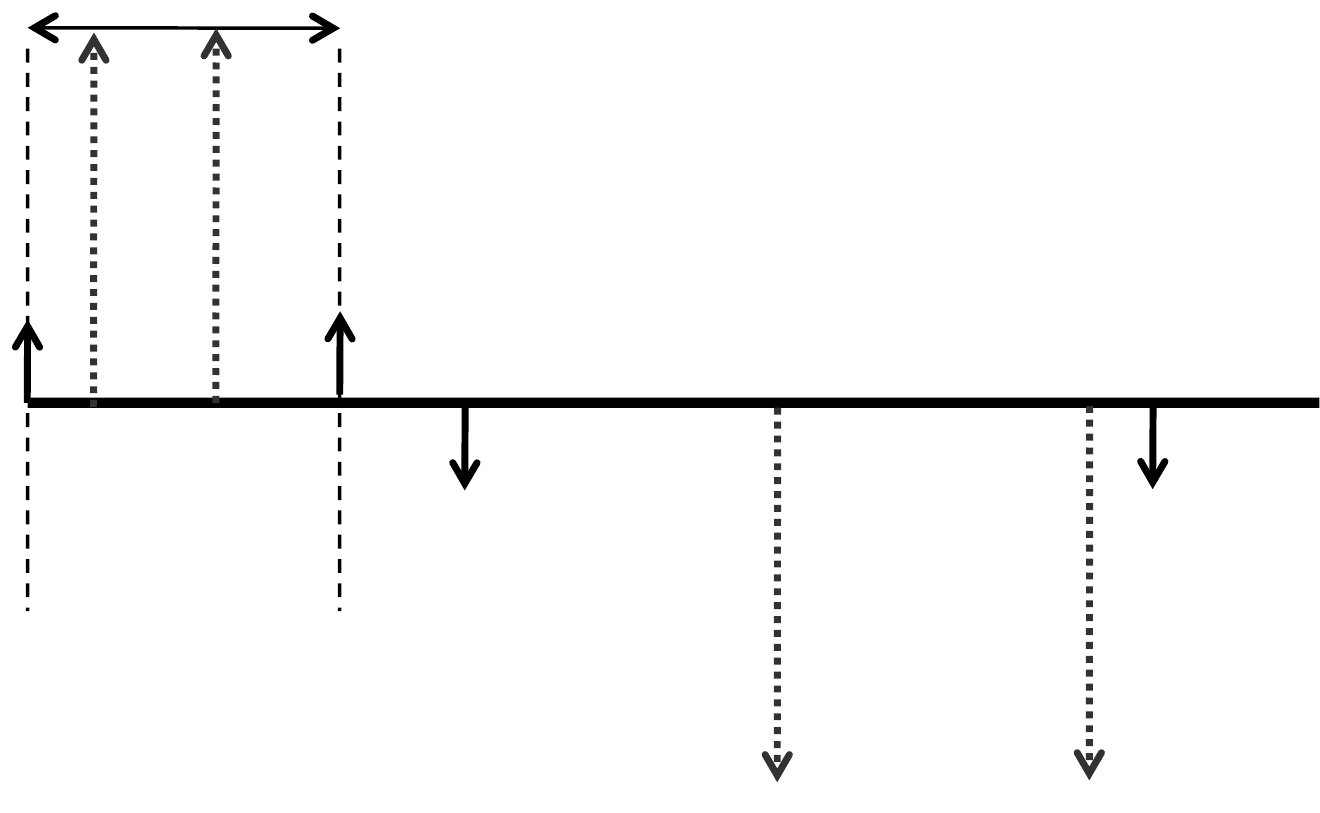}
\put(-7,10){ $\frac{(k-1)c}{\lceil c \rceil}$}
\put(22,10){$\frac{kc}{\lceil c \rceil}$}
\put(30,35){$t^{'}_{k-1}$}
\put(85,35){$t^{'}_{k}$}
\put(8,63){ $\tilde{X}_{k}=2$}
\end{overpic}
}
\caption{This is the scenario when $t_{k-1}^{'}>kc/\lceil c \rceil$. In this scenario, $\tilde{X}_{k}=t_{k}^{'}-t^{'}_{k}-s_{k}$.  In the figure, the solid upward and downward pointing arrows denote the arrival and departure times, respectively, of attacker's jobs, and the dotted arrows represent the jobs from Alice. The size of the arrow is proportional to the size of the job. }
\label{fig:latedepboth}
\end{figure}

 Next we present some remarks about the result derived above. For the attack strategy specified in Theorem \ref{thm:fcfszeroerr}, the attacker incurs 0 error independent of his rate. This means that the attacker can estimate Alice's arrival pattern exactly even when he is permitted to issue jobs at an arbitrarily small rate. This is the case because Alice's jobs are of size 1 and as long as the attacker issues jobs at a frequency greater than 1, he can estimate Alice's arrivals exactly. Note that the size of jobs issued by him can be arbitrarily small. 
 
 As an aside, the result of the theorem can easily be extended to the case when the job sizes of Alice are of arbitrary sizes (for example, when the job sizes from Alice are exponentially distributed). The attack strategy in such a scenario would be the follows. The attacker issues one job every $\delta$ units of time, and the size of each job is $\lambda \delta$. His estimate of Alice's arrivals is the follows. If $t_{k-1}^{'}<t_{k}$ and $t_{k}^{'}=t_{k}+s_{k}$, then he estimates the number of arrivals from Alice between times $t_{k-1}$ and $t_{k}$ to be zero.  If $t_{k-1}^{'}<t_{k}$ and $t_{k}^{'}>t_{k}+s_{k}$, then he estimates the number of arrivals from Alice between times $t_{k-1}$ and $t_{k}$ to be $t_{k}^{'}-(t_{k}+s_{k})$. And finally if $t_{k-1}^{'}>t_{k}$, then he estimates the number of arrivals to be $t_{k}^{'}-s_{k}-t_{k-1}^{'}$. Like before, his estimation is accurate in the third scenario. However, in the first two scenarios, he could incur some error. Nonetheless, the error is bounded by $\delta^{2}$. Therefore, by choosing $\delta$ small enough, the attacker can incur close to zero estimation error. 
 
 This suggests two possible alterations to the FCFS policy which could lead to a higher estimation error for the attacker. First, the policy could limit the frequency at which a user can issue jobs to the system, by imposing a minimum inter-arrival time. This could be done by having a token-bucket filter \cite{tokenbucketfilter}. Second, the scheduler could impose a restriction on the smallest size of a job, perhaps by idling for a short while after serving a small job. While these mechanisms look promising, it is easy to see that they hurt the performance of the system in terms of a smaller throughput region (if the scheduler starts idling) and/or adding to the overall delay. As we shall show in the following section, one can design secure scheduling policies without sacrificing any of the throughput region.

To sum up, FCFS offers very little privacy to its users. On the other hand, TDMA offers the highest levels of privacy, albeit at a high performance penalty (in terms of delay and throughput). In the following section, we propose a parametric scheduling policy which provides a controlled trade-off between delay and privacy.  

\section{Provably secure scheduling policies}\label{sec:securepolicies}

In this section, we derive two policies that are resistant to any attacks that the attacker might possibly launch. The first policy, \emph{accumulate and serve} pre-distorts the arriving traffic before serving them. This erases the fine grained timing information that the attacker can learn. The second policy takes a different approach. TDMA leaks no information of one user's arrivals to the other because slots are statically reserved ahead of time. The second policy, \emph{proportional TDMA}, uses this same principle to guarantee minimal information leakage, however, it is designed so that these reservations alter adaptively at a slow rate so as to minimize delays incurred by the users. We discuss these two policies in the following.

\subsection{Accumulate and Serve Policy}\label{sec:accserve}
The Accumulate and Serve policy is shown in Figure \ref{fig:accpolicy}. In this policy, the scheduler accumulates all the incoming jobs in its buffer for a period of $T$ time units. It then serves all the accumulated jobs from user 1, followed by all the accumulated jobs from user 2, followed by those from user 3, and so on. In the two user scenario discussed before, user 1 could be the attacker, or Alice. If serving all these jobs takes $M$ units of time and if $M<T$, then the scheduler idles for the remaining $T-M$ time before beginning to serve the accumulated jobs. The intuition for such a policy stems from the fact that buffering destroys the timing information of the arrival times of the jobs, thereby reducing the amount of information that can be extracted by the attacker. The attacker can therefore learn, at most, the total number of jobs that arrived from Alice in the accumulate period, and nothing more fine grained. Our subsequent analysis will make use of this fact to show that this policy provides guaranteed levels of privacy to Alice, irrespective of any attack that can be launched by Bob.  Using Foster-Lyapunov stability theorem, it can be shown that the scheduling policy is stable when the sum of rates from all the users is less than 1 (refer Appendix \ref{prf:foster}).

\begin{figure}[t] 
   \centering
   \includegraphics[width=0.7\columnwidth]{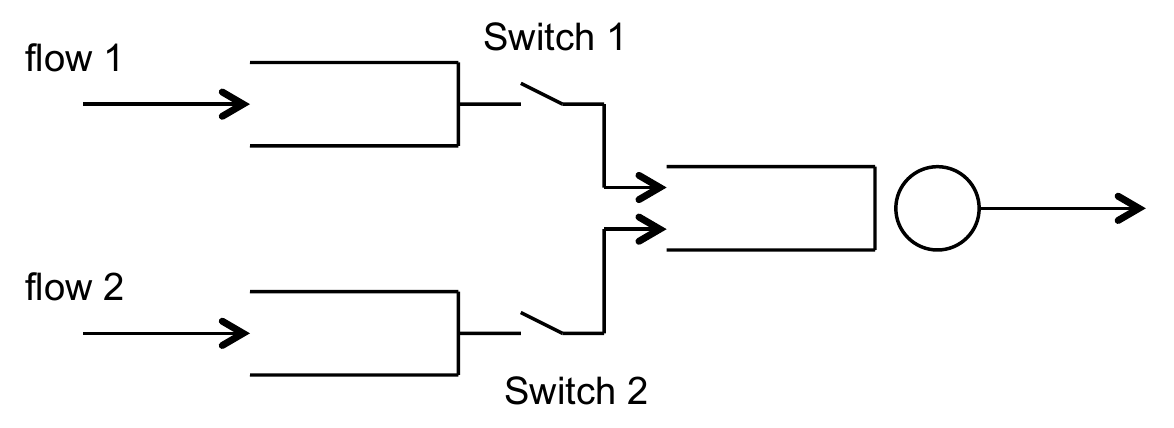} 
   \caption{Pictorial representation of the accumulate and serve policy. The switches close every $T$ time units. All the accumulated jobs from user 1 are served followed by jobs from user 2. The two flows correspond to jobs issued by the two users.}
   \label{fig:accpolicy}
\end{figure}

\subsection{Analysis}

\begin{theorem}\label{thm:accmain} \label{THM:ACCMAIN}
A lower bound on the privacy offered by the accumulate and serve policy is given by, 
\begin{eqnarray}
\mathcal{E}_{\mathrm{Acc Serve}}^{c,\lambda_{2}} & \geq & \lambda_{2}c\left(1-\frac{c}{T}\right)_{+} \nonumber \\
&=& \emax \left(1-\frac{c}{T}\right)_{+}, \label{eq:thmeacc}
\end{eqnarray} 
\end{theorem}

\begin{proof} 
A detailed proof is given in Appendix \ref{app:proofaccmain}. 
 \end{proof}

Some remarks about the consequences of Theorem \ref{thm:accmain}:
\begin{itemize}
\item The accumulate and serve policy guarantees a certain level of privacy regardless of the attacker strategy. Furthermore, if the accumulate period is chosen to be large enough so that $\frac{c}{T} \ll 1$, then the estimation error incurred by the attacker gets increasingly close to the maximum estimation error, which is what he would have incurred by just using the statistics of Alice's arrival. The price paid for choosing a very large value for $T$ is of course the delay experienced by jobs. Delay analysis is carried out in Section \ref{sec:delay}. When $T$ is chosen to be 5 times $c$ or greater, the ratio is greater than 0.8, irrespective of the attacker's strategy. 
\item It is important to note that the error does not depend on the strategy employed by the attacker. In particular, it does not depend on the rate at which the attacker issues his jobs, which seems unusual. This is a consequence of the result in Lemma \ref{lem:observationsuseless}, which says that when the attacker is given minimum additional information (by a genie) about total number of jobs in an accumulate period, his arrival and departure times carry no additional useful information for the purpose of estimation. The rate at which attacker issues his jobs, however, does determine the gap between true estimation error and the genie aided error.  
\item The bound $\mathcal{E}_{\mathrm{Max}}^{c}\left(1-\frac{c}{T}\right)_{+}$ evaluates to 0 when $c \geq T$, which is not very useful. However, we are mostly interested in the scenario when $T$ is chosen to be large enough to guarantee a degree of protection against a given value of $c$.   
\end{itemize}

\subsection{Proportional TDMA (p-TDMA) Policy}
As stated before, TDMA achieves the highest privacy because the times at which the attacker gets served is known to him a priori, he learns nothing by performing the attack. The p-TDMA policy builds upon this idea. In the two user scenario, the policy divides time into slots, and at time 0, the policy statically assigns all odd numbered slots to user 1 and even numbered slots to user 2, just like TDMA. At the end of an \emph{adaptation period} $L$, the policy computes the empirical rates at which the two users have been issuing jobs to the scheduler, $\hat{\lambda}_{1}$ and $\hat{\lambda}_{2}$. Between times $L$ and $2L$, each time slot is reserved for user 1 with probability $\frac{\hat{\lambda}_{1}}{\hat{\lambda}_{1}+\hat{\lambda}_{2}}$, and user 2 with probability $\frac{\hat{\lambda}_{2}}{\hat{\lambda}_{1}+\hat{\lambda}_{2}}$. At time $2L$, the scheduler recomputes the empirical rates, and uses the new rates to decide whether to reserve a slot to serve user 1 or user 2. The policy continues to update the empirical rates every $L$ time units. The idea is that the policy eventually learns the true rates at which the users issue jobs, and allocates times in proportion to these rates, thus reducing the delays incurred.

\begin{theorem}
\begin{eqnarray}
\eptdma \geq \emax \left( 1-\frac{c}{L} \right)_{+}.
\end{eqnarray}   
\end{theorem}
\begin{proof}
The analysis is very similar to that of the accumulate and serve policy. The policy depends only on the empirical rates at which jobs are issued to the scheduler. If the attacker is given the side information about how many jobs Alice has issued in each adaptation period, his arrival and departure times carry no further information and can be discarded. The result follows. 
\end{proof}
Some remarks about this policy follow:
\begin{itemize}
\item Unlike the accumulate and serve policy that trades off delay for privacy, the p-TDMA policy trades off adaptation time for privacy. Larger the value of the adaptation period $L$, longer it takes for the policy to learn the true arrival rates. \emph{However, for stationary arrival processes, the long run average delay is independent of the duration of the adaptation period}. The reason for this is the follows. Irrespective of the value of $L$, after sufficient number of adaptation periods, the empirical rates converge to the true rates, and stay that way. Therefore, in the long run, the delay experienced by the jobs is independent of $L$. In theory, by choosing a large enough value of $L$, this policy is guaranteed to perform close to TDMA on the privacy metric without the excessive delay penalty. 
\item In a practical system, if the arrivals are not stationary, i.e., the statistical description of the arrival process changes with time, or if a user issues jobs only for a short duration which is insufficient for the policy to learn the true arrival rate from him, the delays incurred by the user can still be large. One would have to choose the value of $L$ wisely to strike a balance between the performance for short flows, and the privacy it offers.  
\end{itemize}

\section{Delay Analysis}\label{sec:delay}
The delay experienced by the jobs would depend on the total number of users of the system and the specific patterns of arrivals from them. For purposes of delay comparison, we will consider a scenario where all the users of the system have time invariant arrival statistics. This would be the delay incurred when all users are legitimate innocuous users. This is a reasonable assumption, as during the `normal' mode of scheduler operation, when there are only legitimate users using the scheduler, the policy should offer as minimum delay to the jobs as possible. We will consider the case where there are $M$ users of the system, traffic from user $i$ being a Poisson traffic with parameter $\lambda_{i}$, size of all the jobs being equal to 1 time unit. Define $\lambda=\sum_{i=1}^{M}\lambda_{i}$. We will compute the average delay of a job (averaged over all the jobs from all the users) for each of the policies discussed here.

\begin{theorem} \label{thm:delay}\label{THM:DELAY}
The mean delay experienced by a job for FCFS, TDMA, accumulate and serve and p-TDMA scheduling policies are as follows:
\begin{align} 
\mathcal{D}_{\mathrm{FCFS}} =&1+\frac{\lambda}{2(1-\lambda)}, \label{eq:dfcfs} \\
\mathcal{D}_{\textrm{TDMA}} =& 1+ \frac{M}{2} + \sum_{i=1}^{M} \frac{\lambda_{i}}{\lambda} \frac{\lambda_{i}M^{2}}{2(1-\lambda_{i}M)},\label{eq:dtdma} \\
\mathcal{D}_{\mathrm{Acc\ Serve}} \leq& 1 + \frac{\lambda (T+1)}{2} + \frac{\lambda + T(1-\lambda)^{2}}{2(1-\lambda)}\mathbf{I}_{\{\lambda>\lambda^{*}_{T}\}} \nonumber \\
&+\sqrt{\lambda T}\mathbf{I}_{\{\lambda<\lambda^{*}_{T}\}}, \label{eq:daccserve} \\
\mathcal{D}_{\textrm{p-TDMA}} =& 1+\frac{1}{2(1-\lambda)} + \frac{M-1}{1-\lambda}. \label{eq:dptdma} 
\end{align}
where, $\lambda^{*}_{T}=\frac{2T+1-\sqrt{1+4T}}{2T}$.
\end{theorem}
\begin{proof}
A detailed proof is given in Appendix \ref{app:prfdelay}.
 \end{proof}

\begin{figure}[t] 
   \centering
   \includegraphics[width=\columnwidth]{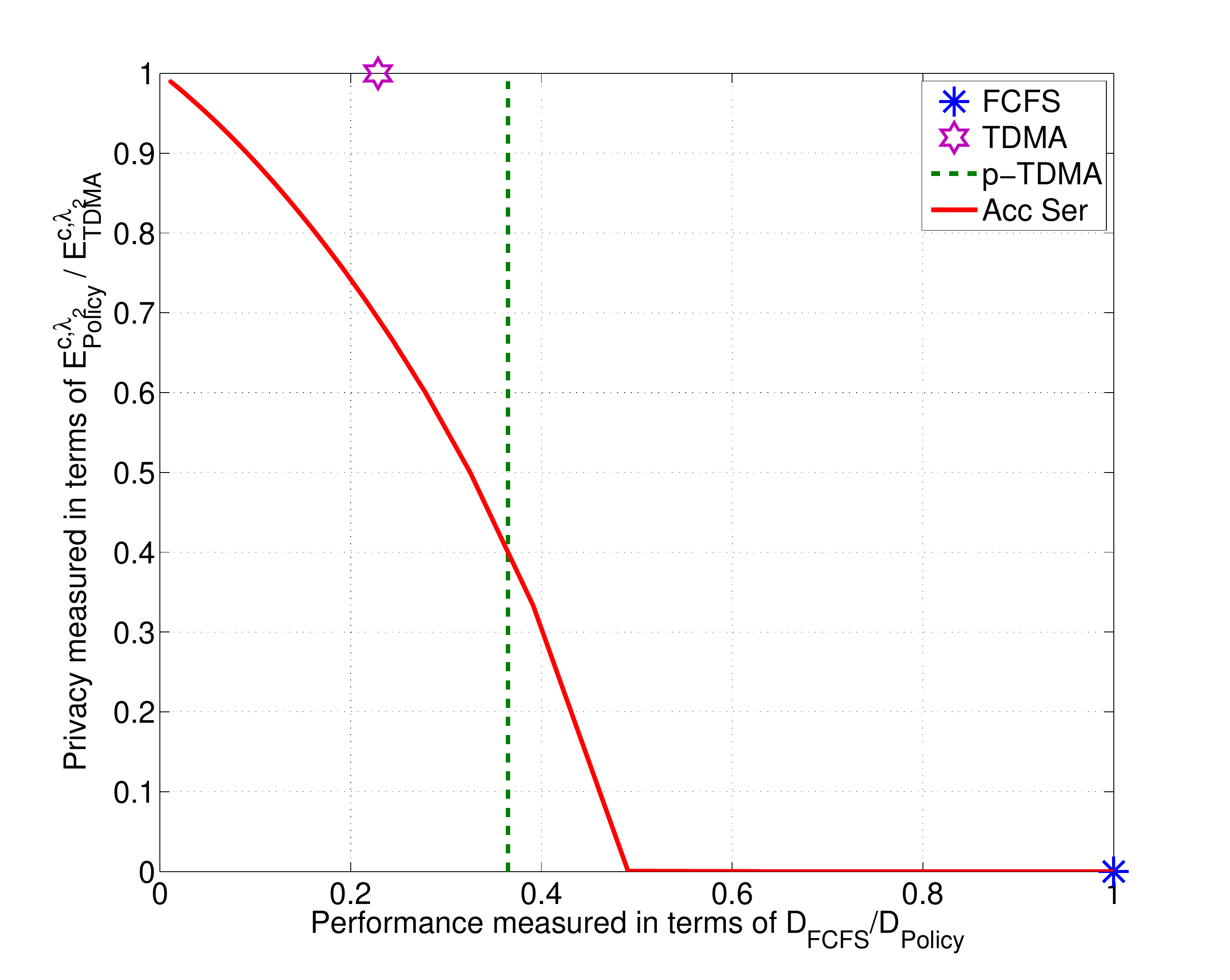} 
   \caption{Privacy-delay trade-off offered by different policies. $\frac{\mathcal{D}_{FCFS}}{\mathcal{D}_{Policy}}$ measures the performance of a policy in terms of delay offered by it, higher the ratio, the better.  $\frac{\mathcal{E}^{c,\lambda_2}_{Policy}}{\mathcal{E}^{c,\lambda_2}_{TDMA}}$ measures the privacy performance of the policy, higher the better. We consider the case when two users use the system, one user issues jobs at a rate 0.2 and the other at 0.45. By varying the parameters associated with each of the two derived policies, we obtain curves for the accumulate and serve and p-TDMA policies. The value of $c$, the clock period, is set to two.}
   \label{fig:privdelay}
\end{figure}

Some remarks about the mean delay computed:
\begin{itemize}
\item For FCFS, accumulate and serve, and p-TDMA, the mean delay is finite, equivalently the system is stable when the sum of the rates from all the users is less than 1, i.e., $\lambda=\sum_{i=1}^{M}\lambda_{i}<1$. However, for TDMA, the mean delay is finite only when the rate from each user is less than $1/M$, i.e., when $\lambda_{i}<1/M$. Therefore, TDMA loses out on multiplexing gains.  
\item $\lim\limits_{\lambda \to 1} \frac{\mathcal{D}_{\mathrm{Acc\ Serve}}}{ \mathcal{D}_{\mathrm{FCFS}}} =1$, i.e., at high traffic loads, the mean delay offered by the accumulate and serve policy is equal to that when FCFS policy is used. In the other extreme, $\lim\limits_{\lambda \to 0} \frac{\mathcal{D}_{\mathrm{Acc\ Serve}}}{ \mathcal{D}_{\mathrm{FCFS}}} = 1+\frac{T}{2}$, which can be quite large especially if $T$ is chosen to be large. For TDMA, the same limit would be $1+\frac{M}{2}$. For $p-TDMA$, it is $1/2+M$. Depending on the values of $M$ and $T$, either TDMA or accumulate and serve might offer the least delays. 
\item In Figure \ref{fig:privdelay}, we plot the privacy offered by each of these policies vs the inverse of the delay offered by them. We consider the case when two users use the scheduler, one of them issues jobs at a rate 0.45~and the other at rate 0.2. FCFS offers the least delay but also the least privacy, while TDMA offers the highest privacy, but at the expense of a significant performance loss. Accumulate and serve, as discussed earlier, is a parametrized policy that can trade-off privacy for delay. On the other hand, the p-TDMA policy can be tuned to achieve highest privacy without affecting the delay performance at all. However, longer the adaptation period, longer the arrival streams from the users have to be for the higher initial delays to average out.
\end{itemize}

\section{Concluding Remarks}

In this work, we discuss the timing side channel that arises naturally when two processes share a common resource. We consider a system where jobs from two users arrive at a processor. The task is to find a scheduling policy for the processor that minimizes information leakage about one user's job arrival pattern to the other while not significantly increasing the average delay seen by the jobs. We use the estimation error to quantify the information leak, and use it to demonstrate that the first come first served scheduling policy leaks significant information. We also propose two scheduling policies that provide guaranteed levels of privacy. Do these policies continue to be secure if the arrivals were a general arrival process instead of being Poisson? To answer this question, note that the operation of both these derived policies depend on the number of arrivals in a batch (in one accumulate period, or one adaptation period). Therefore, as long as the arrival process from Alice is such that, given the total number of arrivals in an accumulate/adaptation period, there is a large uncertainty in the arrival times of each jobs, the derived policies will still fare well.

\bibliographystyle{ieeetr}
\bibliography{sachin,ISITA10bib,negarlist1,Parv,tomsbib,negarlist,nikitaborisov,CCS07,career,hatswitch,tm,timingsidechannelrefs}

\begin{appendices}

\section{Proof  of Theorem \ref{thm:accmain}}\label{app:proofaccmain}

Suppose that $n$ jobs from Bob arrive in $K$ accumulate periods. Let $\{B_{1},B_{2},\ldots,B_{K}\}$, denoted in short by $B_{1}^{K}$, be the number of jobs arriving from Alice in accumulate periods $1,2,\ldots,K$ respectively. $B_{1}^{K}$ is the side information given to the attacker. The following lemma shows that the side information is the maximum information the attacker can hope to learn. 
\begin{lemma}\label{lem:observationsuseless}
If the attacker is given the side information $B_{1}^{K}$, the arrival and completion times of his jobs carry no information and can be discarded. Let $X_{k}$ be the random variable denoting the number of jobs that arrive from Alice in between clock ticks $k-1$ and $k$. The conditional expectation, 
$\mathbf{E}(X_k|t_{1}^{n},t^{'n}_{1},s_{1}^{n},B_{1}^{K})$ is the same as $\mathbf{E}(X_k|B_{1}^{K})$. 
\end{lemma}
\begin{IEEEproof}
Consider the conditional probability distribution $\Pr(X_k|B_{1}^{K},t_{1}^{n},t^{'n}_{1},s_{1}^{n})$. If accumulate and serve policy is used at the scheduler, the departure times of Bob's jobs are just a function of the arrival times and the size of his jobs, and the sequence $\{B_1,B_2,\ldots,B_K\}$. Therefore,
\begin{eqnarray}
\lefteqn{\Pr(X_k|B_{1}^{K},t_{1}^{n},t_{1}^{'n},s_{1}^{n})} \nonumber \\
 & =& \frac{\Pr(X_k,B_{1}^{K},t_{1}^{n},t_{1}^{'n},s_{1}^{n})}{\Pr(B_{1}^{K},t_{1}^{n},t_{1}^{'n},s_{1}^{n})} \nonumber \\
&=&  \frac{\Pr(X_k,B_{1}^{K},t_{1}^{n},s_{1}^{n})}{\Pr(B_{1}^{K},t_{1}^{n},s_{1}^{n})}  \frac{\Pr(t_{1}^{'n}|X_k,B_{1}^{K},t_{1}^{n},s_{1}^{n})}{\Pr(t_{1}^{'n}|B_{1}^{K},t_{1}^{n},s_{1}^{n})} \nonumber \\
&=&  \Pr(X_k|B_{1}^{K},t_{1}^{n},s_{1}^{n}) \label{eq:bgivesall},
\end{eqnarray}
where \eqref{eq:bgivesall} follows because $t_{1}^{'n}$ is a function of $t_{1}^{n},s_{1}^{n}$ and $B_{1}^{K}$. As a consequence, the three sets of random variables $\{\arrtimes,\sizes\}-B_{1}^{K}-\{X_k\}\}$ form a Markov chain. Thus, $\Pr(X_k|B_{1}^{K},t_{1}^{n},s_{1}^{n})=\Pr(X_k|B_{1}^{K}).$
\end{IEEEproof}

As a result of Lemma \ref{lem:observationsuseless}, if the attacker is given $B_{1}^{K}$, the conditional distribution of $X_k$, and hence the resulting estimation error he will incur, is independent of the timing and number of jobs that the attacker sends. Also, if the resulting estimation error is denoted by  $\mathcal{E}_{\mathrm{Acc Serve}}^{c,B_{1}^{K}}$, then, from \eqref{eq:inequalitiesacc}, 
\begin{equation}
\mathcal{E}_{\mathrm{Acc Serve}}^{c,\lambda}\geq \mathcal{E}_{\mathrm{Acc Serve}}^{c,B_{1}^{K}}. \label{eq:eacclwr}
\end{equation} 
This is because, we are giving the attacker extra information which is not directly available to him. With this information, the attacker can only be better off in his estimation procedure. Hence, $\mathcal{E}_{\mathrm{Acc Serve}}^{c,B_{1}^{K}}$ bounds the estimation error that the attacker incurs, and we will shortly show that this is large enough. Next, we turn our attention to computing $\mathcal{E}_{\mathrm{Acc Serve}}^{c,B_{1}^{K}}$.

We first consider the case when the following hold:
\begin{description}
\item[\textbf{C1}] The accumulate period, $T$, is an integer multiple of $c$, the duration of a clock period.
\item[\textbf{C2}]  At time 0, the first clock tick aligns with the start of the first accumulate period.
\end{description}
If conditions \textbf{C1} and \textbf{C2} are met, every clock period is then completely contained within an accumulate period. 

Suppose that the clock period $k$ is contained within the accumulate period $m$. Let $X_k$ denote the number of arrivals from Alice in clock period $k$, and let $\overline{X}_k$ denote the number of arrivals from Alice in the other clock periods in accumulate period $m$. 
Then, it can be shown that, conditioned on the total number of arrivals in accumulate period $m$, $B_{m}$, $X_{k}$ is independent of the arrivals in the other accumulate periods. Therefore, $\Pr(X_k|B_{1}^{K})=\Pr(X_k|B_m)$. For a Poisson process, conditioned on the total number of arrivals in a given period, the arrival times themselves are uniformly distributed in the period \cite{Gallagernotes}. Therefore, conditioned on $B_{m}=\beta_{m}$, the distribution of $X_{k}$ is a Binomial random variable with parameters $(\beta_{m},\frac{c}{T})$. Given the value of $\beta_{m}$, the best estimate for $X_{k}$ is then $\frac{c}{T}\beta_{m}$. Thus, the estimation error incurred by the attacker in a slot in which $\beta_m$ jobs arrived is equal to the variance of the Binomial random variable equal to $\mathcal{E}_{\beta_{m}}=\beta_{m}\frac{c}{T}\left(1-\frac{c}{T}\right)$. Given that $\beta_m$ itself is a random variable, we need to average over its distribution to get the estimation error incurred by the attacker. $\beta_m$ is  Poisson with mean $\lambda_{2}T$, and the estimation error in a typical clock period is
\begin{eqnarray}
 \mathbf{E}_{\beta_m}[\mathcal{E}_{\beta_m}] &=& \mathbf{E}_{\beta_m}[\beta_{m}\frac{c}{T}\left(1-\frac{c}{T}\right)]\nonumber \\
&=& \lambda_{2}c \left(1-\frac{c}{T}\right) \nonumber \\
&=& \mathcal{E}_{Max}^{c} \left(1-\frac{c}{T} \right), \label{eq:eacc}
\end{eqnarray}
where \eqref{eq:eacc} follows from \eqref{eq:unifub}. 

When the conditions \textbf{C1} and \textbf{C2} do not hold, it can be shown the estimation error incurred by the attacker is greater than $\mathcal{E}_{Max}^{c} \left(1-\frac{c}{T} \right)$. It can be shown that the empirical average of the estimation error incurred in a clock period, as given in \eqref{eq:eacc} is also equal to the time average of the estimation error incurred by the attacker, which, as defined in \eqref{eq:privacy} is the privacy metric we are interested in. Therefore, $\mathcal{E}^{c,B_{1}^{K}}_{\textrm{Acc\ Serve}} = \mathcal{E}_{Max} \left(1-\frac{c}{T} \right)$.

\paragraph*{The case when $c>T$}
If $c>T$, then the side information $B_{1}^{K}$ of the attacker is finer than the information which he wishes to extract from Alice's arrival process. Therefore, in such a case, the estimation error incurred by the attacker is equal to zero. Recall that in reality the attacker does not have this side information, and he has to infer it based on the arrival and departure times of his jobs. We are not interested in analyzing this case as we expect that in a system design, the accumulate period will be chosen to be sufficiently larger than $c$.   \hfill 

\section{Stability of the accumulate and serve policy}\label{prf:foster}
Let $Q_{n}$ be the number of jobs waiting to be served at the beginning of the $n^{th}$ accumulate period, i.e., at time $nT^{-}$. Because $T$ is an integer multiple of the service time, note that there are no jobs with partially completed service at this time, and therefore the total work in the system at the end of accumulate period is always an integer. Let $A_{n}$ denote the total number of jobs that arrive from all the users in the $n^{th}$ accumulate period. Then, $Q_{n}$ is a Markov chain with the state space $\mathcal{Z}^{+}$, and state update equation given by 
\begin{equation}
Q_{n+1}=(Q_{n}+A_{n}-T)_{+} \label{eqn:queueacc},
\end{equation}
 where the notation $(i)_{+}$ stands for $\max\{0,i\}$. This is because, along with the jobs already waiting to be served at time $nT^{-}$, $A_{n}$ number of jobs arrive. Out of these, at most $T$ of them get served in the accumulate period. The stability of this queue also implies that the mean waiting times of the jobs is always bounded as long as the sum of arrival rates from all the users is less than the service rate.
 
For a state $q\in\mathcal{Z}^{+}$, define the  Lyapunov function to be $V(q)=q^{2}/2$. We will use Foster-Lyapunov's theorem to show that the Markov chain is positive recurrent. Let $\lambda$ be the sum of the arrival rates. Consider the case when each of the arrival process is a Poisson process, then the cumulative arrival process is a Poisson process as well. $A_{n}$ is then a Poisson random variable with parameter $\lambda T$. Given $Q_{n}=q$, the expected value of the drift is given by 
\begin{eqnarray}
\lefteqn{\mathbf{E}[V(Q_{n+1})-V(Q_{n})|Q_{n}=q]} \nonumber\\
&=& \mathbf{E}\left[\frac{((q+A_{n}-T)_{+})^{2}-q^{2}}{2}\right] \nonumber \\
&\leq& \frac{1}{2} \mathbf{E}\left[ (q+A_{n}-\alpha)^{2} -q^{2} \right], \ \alpha\leq T  \label{eq:squarebigger} \\
&=&\frac{1}{2} \Bigg( (q-\alpha)^{2}+\lambda T + (\lambda T)^{2} +2\lambda T( q-\alpha) - q^{2 } \Bigg) \nonumber\\
&=& \underbrace{\frac{\lambda T+(\alpha-\lambda T)^{2}}{2}}_{K_{1}} -\underbrace{(\alpha-\lambda T)}_{K_{2}}q,
\end{eqnarray}
where \eqref{eq:squarebigger} follows because, for real numbers $i,\alpha,\beta$, $((i-\beta)_{+})^{2}\leq (i-\alpha)^{2}, \forall \alpha \leq \beta $. By setting $\alpha$ to $T$, the condition for the system to be stable translates to $\lambda<1$. This is because, if $\lambda<1$, then $K_2>0$, and therefore the expected drift of the Lyapunov function is negative for $q$ large enough.
\hfill 

Furthermore, one can bound the mean number of jobs in the queue in the steady state by
\begin{eqnarray}
\mathbf{E}[Q_{n}] &\leq & \frac{K_{1}}{K_{2}} =\frac{\lambda T+(\alpha-\lambda T)^{2}}{2(\alpha-\lambda T)}. \label{eqn:qbound}
\end{eqnarray}
The bound in \eqref{eqn:qbound} holds for all values of $\alpha$ in the range $(\lambda T, T]$. Choosing the value of $\alpha$ that results in the tightest bound, it can be shown that 
\begin{eqnarray}
\mathbf{E}[Q_{n}] &\leq & \begin{cases}
\sqrt{\lambda T} & \mbox{\ if\ } \lambda< \lambda^{*}_{T} \\
\frac{\lambda +(1-\lambda)^{2}}{2(1-\lambda)} & \mbox{\ if \ } \lambda \geq \lambda^{*}_{T} 
\end{cases}, \label{eqn:qtightbound}
\end{eqnarray}
 where $\lambda^{*}_{T}=\frac{2T+1-\sqrt{1+4T}}{2T}$.

\section{Proof of Theorem \ref{THM:DELAY}}\label{app:prfdelay}
The FCFS system is a simple $M/D/1$ queue, and the mean delay can be derived from  Polleczek and Khinchine's formula for the $M/G/1$ system, refer \cite{renewal}. For the proof for TDMA, refer \cite{tdmadelay}. The proof for accumulate and serve and the p-TDMA policies follow.

\subsection{Mean delay of accumulate and serve policy}
The arriving jobs are stored in the buffer till the start of the next accumulate period. A job, chosen at random gets delayed by $T/2$ time units at this stage. From the point of view of the processor, the inputs to the queue are batch arrivals which arrive every $T$ time units, and the size of a batch is a Poisson random variable. The average delay across all jobs does not depend on the order in which jobs are served by the server, as long as the server does not idle. For delay analysis, we can therefore assume that the server serves these jobs in FCFS manner. In reality, recall that the server serves all the jobs from one user followed by all the jobs from the other user.

Let $A_{k}$ denote the size of the batch that arrives at the end of the accumulate period $k$ and $Q_{k}$ be the number of jobs that are not yet served at the end of period $k$. By the FCFS assumption, $Q_{k}$ is also the time the arriving batch waits before it gets served. The queue update equation is given in \eqref{eqn:queueacc}, and consequently, the bounds derived in Appendix \ref{prf:foster} hold. Therefore, the mean waiting time before a batch of jobs starts getting served is upper bounded by $ \frac{\lambda +(1-\lambda)^{2}}{2(1-\lambda)}  \mathbf{I}_{\{\lambda\geq\lambda^{*}_{T}\}}+\sqrt{\lambda T}\mathbf{I}_{\{\lambda<\lambda^{*}_{T}\}}$. 

In order to compute the mean delay experienced by a job, first note that a job drawn at random from the first $K$ jobs, where $K$ is a large number, has a higher chance of belonging to a bigger batch than a smaller one. In fact, the probability that a job drawn at random belongs to a batch of size $b$ is given by $\frac{b \mathbf{P}(A_{k}=b)}{\mathbf{E}[A_{k}]}$. Given that the job is from a batch of size $b$, the mean number of jobs ahead of it can be shown to be  $\frac{b-1}{2}$, consequently, conditional on the batch size to be $b$, the job has to wait for an additional time of $\frac{b-1}{2}$ time units after the batch starts service before it can get served. Averaging over $b$, we can get the additional waiting time of a jobs drawn at random before it gets served to be
\begin{eqnarray}
\sum_{b} b \frac{{\Pr}(A_{k}=b)}{\mathbf{E}[A_{k}]} \left(\frac{b-1}{2}\right)= \frac{\mathbf{E}[A_{k}^{2}]-\mathbf{E}[A_{k}]}{2\mathbf{E}[A_{k}]} = \frac{\lambda T}{2}
\end{eqnarray} 

Once it gets to service, the service time of the job is a fixed 1 time unit. Therefore, the delay experienced by a job drawn at random, which is the mean delay of offered by the scheduling policy is bounded by \eqref{eq:daccserve}.

\subsection{Mean delay of p-TDMA policy}
If the arrival process from each user is a Poisson process of rate $\lambda_{i}$, in the steady state, the empirical arrival rates converge to the true arrival rates. User $i$ gets service in a time slot with probability $\lambda_{i}/\lambda$. Let $Q_{n}$ be the number of unserved jobs from user $i$ at the beginning of time slot $n$. We then have the following queuing equation,
$$ Q_{n+1} = \left( Q_{n} - D_{n} \right)_{+} + A_{n} , $$
where $A_{n}$ denotes the number of job arrivals between times $n$ and $n+1$, which is a Poisson random variable with mean $\lambda_{i}$, and $D_{n}$ is a Bernoulli random variable which takes the value 1 if user $i$ gets served in time slot $n$, which happens with probability $\lambda_{i}/\lambda$, or 0 otherwise. The queuing equation holds because, the system serves only jobs that have already arrived by time $n$ in time slot $n$.

In the steady state, both $Q_{n}$ and $Q_{n+1}$ have the same distribution. Let $q_{i}\doteq \Pr(Q_{n}=i)=\Pr(Q_{n+1}=i)$, and $\mathcal{Q}(z) \doteq \sum\limits_{i=0}^{\infty} q_{i}z^{i}$ be the Z-transform of the steady state distribution. We then have
\begin{eqnarray}
\mathcal{Q}(z) &=& \mathbf{E}[z^{Q_{n+1}}] = \mathbf{E}[z^{ \left( Q_{n} - D_{n} \right)_{+}+ A_{n}}] \label{eq:followsfromqueuing}\\
&=& q_{0} \mathbf{E}[z^{A_{n}}] \left( q_{0} \mathbf{E}[z^{(-D_{n})_{+}}] + \sum\limits_{i=1}^{\infty} q_{i} \mathbf{E} [z^{i-D_{n}}] \right) \nonumber \\
&=& \frac{q_{0}(\lambda_{i}/\lambda)e^{\lambda_{i}(z-1)}(1-z^{-1})}{1-(1-(\lambda_{i}/\lambda)+(\lambda_{i}/\lambda)z^{-1})e^{\lambda_{i}(z-1)}},
\end{eqnarray}
where \eqref{eq:followsfromqueuing} follows from the queuing equation. Use the fact that $\lim\limits_{z\to 1} \mathcal{Q}(z)=1$ to solve for $q_{0}=(1-\lambda)$. The mean number of unserved jobs from user $i$ in the steady state is then equal to $\sum\limits_{i=0}^{\infty}i q_{i}=\lim \limits_{z\to1} \dot{\mathcal{Q}}(z),$ which can be computed. Using Little's law to relate the average queue length to the mean delay experienced by a job from user $i$, and averaging it across all users of the system, we get the result given in equation \eqref{eq:dptdma}. The result also includes the time $1/2$, the mean time an arriving job waits before the nearest multiple of one time unit.
\end{appendices}

\end{document}